\documentclass[journal,draftcls,onecolumn,12pt,twoside]{IEEEtranTCOM}
\renewcommand{\baselinestretch}{1.7}

\usepackage{moreverb}
\usepackage{epsfig}
\usepackage{amsmath,amssymb,amsthm,mathrsfs,amsfonts,dsfont}
\usepackage{adjustbox,lipsum}
\usepackage[linesnumbered,ruled,vlined]{algorithm2e}
\usepackage{amsfonts}
\usepackage{epsfig}
\usepackage{amssymb}
\usepackage{amsmath}
\usepackage{amsthm}
\usepackage{multirow}
\usepackage{setspace}
\usepackage{rotating}
\usepackage{graphicx}
\usepackage{tabularx}
\usepackage{array}
\usepackage{anyfontsize}
\usepackage{color,soul}
\usepackage{graphicx,dblfloatfix}
\usepackage{epstopdf}
\usepackage{blindtext}
\usepackage{amsmath}
\usepackage{amsthm,amssymb,amsmath,bm}
\usepackage{subfigure}
\usepackage{amsfonts}
\usepackage{epsfig}
\usepackage{amssymb}
\usepackage{amsmath}
\usepackage{cite}
\usepackage{graphicx}
\usepackage{fancyhdr}
\usepackage[font={small}]{caption}
\usepackage{tabularx}
\usepackage{tcolorbox}
\usepackage{cite}
\usepackage{setspace}
\let\labelindent\relax
\usepackage[inline]{enumitem}

\usepackage{amsthm}

\newtheorem{theorem}{Theorem}
\newtheorem{lemma}{Lemma}

\newtheorem{definition}{Definition}
\newtheorem{remark}{Remark}
\newtheorem{proposition}{Proposition}

\allowdisplaybreaks
\title{On the Semantic Security in the General Bounded Storage Model: A New Proof
} 

\author{ 
\IEEEauthorblockN{Mohammad~Moltafet, Hamid~R.~Sadjadpour, and Zouheir~Rezki
}
\thanks{The authors are with the Department of Electrical and Computer Engineering, University of California Santa Cruz (UCSC), Santa Cruz, CA 95064, USA (e-mail: mmoltafe@ucsc.edu, hamid@soe.ucsc.edu, zrezki@ucsc.edu).
}
}

\begin{document}
\maketitle
\sloppy 

\begin{abstract}
In the bounded storage model introduced by Maurer, the adversary is computationally unbounded and has a bounded storage capacity. In this model, \textit{information-theoretic} secrecy is guaranteed by using a publicly available random string whose length is larger than the adversary storage capacity. The protocol proposed by Maurer is simple, from the perspective of implementation, and efficient,  from the perspective of the initial secret key size and random string length. However, he provided the proof of the security for the case where the adversary can access a constant fraction of the random string and store only \textit{original bits} of the random string. In this paper, we provide a new proof of the security of the protocol proposed by Maurer for the \textit{general} bounded storage model, i.e., the adversary can access all bits of the random string, and store the output of any Boolean function on the string. We reaffirm that the protocol is \textit{absolutely semantically secure} in the general bounded storage model.

\emph{Index Terms--} Perfect security,  information-theoretic secrecy, bounded storage model.
\end{abstract}

\section{Introduction}
The most effective approach of secrecy coding that provides \textit{perfect secrecy} is the one-time pad (or Vernam cipher). Vernam introduced his cipher in 1926 \cite{Vernam1926} and Shannon in 1949 \cite{Shannon} proved that the one-time pad scheme provides perfect secrecy. According to the one-time pad scheme, to securely exchange a message, an independent and uniformly distributed one-time pad whose size equals the message length needs to be established between the transmitter and receiver. More specifically, in \cite{Shannon}, Shannon proved that perfect secrecy against an all-powerful adversary with unbounded storage capacity and unbounded computational power who has complete access to the communication line is only achievable if the uncertainty
of the secret key is at least as great as that of the plaintext, this is also known as {\it Shannon impossibility result}\cite{Yevgeniy2012}.

In 1992, Maurer introduced the bounded storage model and proposed the first protocol in the bounded storage model that provides \textit{information-theoretic} secrecy in the seminal work \cite{Maurer}. In the bounded storage model, the adversary is computationally unbounded and has a bounded storage capacity, and information-theoretic secrecy is guaranteed by using a publicly available random string whose length is larger than the adversary storage capacity and message size. In this model, the adversary performs an attack in two phases. In phase one, first, the transmitter and receiver establish a secret key, and then, the random string is broadcast. Using the shared secret key, an encryption protocol, and the random string, the transmitter and receiver compute a final key to encrypt and decrypt the message. In this phase, the adversary can compute a function on the random string and store the result. In the second phase, the random string is not available and the adversary is provided with the ciphertext,  \textit{the secret key},  \textit{unbounded storage capacity}, and unbounded computational power. He tries to get information about the encrypted message using the provided information.

Security in the bounded storage model is directly related to the size of the public random string; the larger the random string, the more secure the system. Let $k$ denote the security parameter which determines the length of the random string, $m$ denote the message length, and $n$ denote a large positive integer.
In \cite{Maurer}, the size of the secret key is $k\log_2{n}$, the length of the random string is $kn$, and the author provided the proof of the security for the case where the adversary can access a constant fraction of the random string and store only \textit{original bits} of the random string. Until 1997, it was an open problem to achieve information-theoretic secrecy in a \textit{general} bounded storage model, where the adversary can access all bits of the random string, and store the output of any Boolean function on the string.  
The authors of \cite{Maurergeneralb} provided a protocol and proved that it provides information-theoretic secrecy in the general bounded storage model.  
However, the protocol is much more complicated than the protocol provided in  \cite{Maurer} and the security results are not as efficient as desired. It requires that the transmitter and receiver transmit and store a considerable number of bits. In addition, the protocol is implemented by using multiplication in the field $F_{2^l}$, which is costly for the required amount of transmission. More specifically, to make sure that the probability of revealing information about the plaintext to the adversary is smaller than $\tau$, the protocol requires Alice and Bob to transmit ${3}/{\tau^2}$ bits and store ${3}/{\tau^2}\log |\boldsymbol{\alpha}|$ bits, where $|\boldsymbol{\alpha}|$ is the length of the random string. For example, if we consider $|\boldsymbol{\alpha}|= 2^{37}$, and we require $\tau=10^{-6}$, Alice and Bob have to transmit $3\times10^{12}$ bits, store $1.11\times10^{14}$ bits, and the protocol requires multiplications of element in the field $F_{2^{3\times10^{12}}}$.
The authors of \cite{AumannRabinAR} provided a simple protocol (distinct from the one introduced in \cite{Maurer}) that is provably secure in the general bounded storage model. The protocol exploits a secret key of length $k\log_2{n}$
and requires a random string with length $mkn$ bits which is a considerable number of random bits compared to the protocol introduced in \cite{Maurer}.
In \cite{AumannZongRabin02}, the authors extended the work in \cite{AumannRabinAR} and provided a new provably secure protocol in which the size of the random string is $n$  which is shorter than the previous one.
However, the main problem with this protocol is that the size of the secret key is $mk\log_2{n}$, 
which is much longer than the message length. In \cite{DingRabin02}, the authors proved that by using the protocol provided in \cite{AumannZongRabin02}, the shared secret key can be used to securely transmit an exponential number of messages against an adaptive attacker, i.e., the attacker can adaptively learn the final keys.
The work in \cite{DziembowskiMaurer02} is the first work that provided the proof of the security for the protocol provided in \cite{Maurer} in the general bounded storage model with a secret key of size $k\log_2{n}$ and a random string of size $k(n+m-1)$. They proved that the statistical distance between the final key and the uniform distribution is very small.

In this paper, we provide a new proof for the security of the provided protocol by Maurer in \cite{Maurer} 
in the \textit{general} bounded storage model with a secret key of size $k\log_2{n}$ and a random string of size $kn$ (and thus, smaller than that required in \cite{DziembowskiMaurer02}).  We reaffirm that the protocol is \textit{absolutely semantically secure} in the general bounded storage model. 
The proof is different from (and simpler than) the approach used in \cite{DziembowskiMaurer02}. The main idea behind the proof is as follows. First, we demonstrate that if the adversary is provided with all but one bit of the plaintext, the probability that he can compute the missing bit is exponentially small in the security parameter $k$, which is called \textit{bit security}. Next, we establish the relationship between bit security and semantic security. Specifically, we illustrate that if the adversary can compromise the semantic security of the protocol, he can compute the missing bit, thus contradicting the bit security. In the proof of bit security, our main approach is to demonstrate that the number of strings for which an arbitrary decoding function of the adversary in Phase II can compute the missing bit is very small compared to the number of random strings resulting in the same output as in the first phase of the attack.

 \subsection{Organization}
The paper is organized as follows. The encryption and decryption protocol and main results of the paper are presented in Section~\ref{Protocol and Main Results}.   In Section~\ref{Proof bounded Storage Model}, the proof of security in the general bounded storage model is presented. Finally, concluding remarks are made in Section~\ref{Conclusions}.

\subsection{Notation}  
A random vector is denoted by a bold capital letter,  whereas the corresponding underlined capital letter denotes a realization of the random vector; a random variable is denoted by a capital letter, whereas the corresponding small letter denotes a realization of the random variable; and a set is denoted by a calligraphy letter. Let $\mathcal{G}$ be a finite set, then, $G\stackrel{R}{\leftarrow} \mathcal{G}$ denotes choosing $G$ uniformly at random from $\mathcal{G}$. All the logarithm functions in this paper have base $2$.

\section{The Protocol and Main Results}\label{Protocol and Main Results}
Let ${\boldsymbol{\alpha}=(\boldsymbol{\alpha}^{(1)},\ldots,\boldsymbol{\alpha}^{(k)}),}$ where $\boldsymbol{\alpha}^{(j)}\in\{0,1\}^n$ for all $1\le j\le k,$ denote a random string with length $kn$ and $\boldsymbol{\alpha}^{(j)}[i]$ denote the $(i+1)$th element of string $\boldsymbol{\alpha}^{(j)}$, i.e., the bits in string $\boldsymbol{\alpha}^{(j)}$ are indexed from $0$ to $n-1$. The bits of the random string $\boldsymbol\alpha$ are uniformly distributed and statistically independent. Let $\mathbb{Z}_n=\{0,1,\ldots,n-1\}$ denote a group with addition modulo $n$ as the group operation which is shown by $\pmb{+}$. Next, we present the encryption protocol.

\subsection{The Protocol}
Suppose that Alice wants to send message ${{\mathbf{M}=(M_1,\ldots,M_m)\in\{0,1\}^{m}}}$ to Bob in the presence of an adversary whose storage capacity is bounded by $\beta=\gamma kn$, with $\gamma<1$. Note that the message size is smaller than $n$. The main goal is to provide the final key ${\mathbf{X}=(X_1,\ldots,X_m)\in\{0,1\}^{m}}$ for encryption and decryption of the message $\mathbf{M}$. In this regard, first, they establish a secret key $\mathbf{Z}=(Z_1,\ldots,Z_k)$
such that $\mathbf{Z}\stackrel{R}{\leftarrow} \mathbb{Z}_n^k$, and thus of size $|\mathbf{Z}|=k\log n$ bits.
Then, using the shared key $\mathbf{Z}$, they compute the tuple of sub-keys $\mathcal{S}$ containing $m$ sub-keys, denoted as ${\mathcal{S}=\left(\mathbf{S}^{(1)},\ldots,\mathbf{S}^{(m)}\right)}$,
where $\mathbf{S}^{(i)}$ is computed as follows:
\begin{align}\nonumber 
\mathbf{S}^{(i)}&=
\mathbf{Z}\pmb{+}(i-1)\bm{1}\\&\label{eq1}=(Z_1\pmb{+}(i-1),\ldots,Z_k\pmb{+}(i-1)),
\end{align}
where $\bm{1}=(1,\ldots,1)\in \mathbb{Z}_n^k$. The random string $\boldsymbol\alpha$ is publicly available and Alice and Bob observe it on the fly and by use of the set of sub-keys compute the $i^{th}$ bit of the final key $\mathbf{X}$, i.e., $X_i$ for all $1\le i \le m$,  as follows:
\begin{align}\label{eq2ql}
 X_i
=\bigoplus_{j=1}^k\boldsymbol{\alpha}^{(j)}[S^{(i)}_{j}],
\end{align}
where $S^{(i)}_{j}=Z_{j}\pmb{+}(i-1)$, i.e., the $j$th element of the sub-key $\mathbf{S}^{(i)}$.
Now, Alice computes $\mathbf{C}=\mathbf{M}\oplus \mathbf{X}$, where $\oplus$ denotes bit-wise XOR, and sends $\mathbf{C}$ to Bob as the ciphertext, and Bob decrypts the message by computing $\mathbf{M}=\mathbf{C}\oplus \mathbf{X}$. The steps of the encryption and decryption procedure are summarized in Algorithm~\ref{Alg_Prot}.

\begin{algorithm}[t]
\SetAlgoLined
Message: $\mathbf{M}=(M_1,\ldots,M_m)\in \{0,1\}^m$,

Secret key: $\mathbf{Z}=(Z_1,\ldots,Z_k)$
such that $\mathbf{Z}\stackrel{R}{\leftarrow} \mathbb{Z}_n^k$,

Random string: $\boldsymbol\alpha=(\boldsymbol\alpha^{(1)},\ldots,\boldsymbol\alpha^{(k)}),~\boldsymbol\alpha^{(j)}\in\{0,1\}^n,~\forall 1\le j\le k$,

\For  {$i=1$ to $m$}{
Compute sub-key $\mathbf{S}^{(i)}$:

$\mathbf{S}^{(i)}=
(Z_1\pmb{+}(i-1),\ldots,Z_k\pmb{+}(i-1)),$

Compute the $i$th bit of the final key $\mathbf{X}$, $X_i$:

$X_i
=\bigoplus_{j=1}^k\boldsymbol\alpha^{(j)}[S^{(i)}_{j}]$,

}

Alice and Bob set the final key $\mathbf{X}=(X_1,\ldots,X_m)$,

Alice encrypts $\mathbf{C}=\mathbf{M}\oplus \mathbf{X}$, and sends $\mathbf{C}$ to Bob,

Bob decrypts  the message $\mathbf{M}=\mathbf{C}\oplus \mathbf{X}$.

\caption{ Encryption and decryption protocol of message $\mathbf{M}$}
\label{Alg_Prot}
\end{algorithm} 
\begin{remark}\label{rem-01}
To implement the encryption and decryption protocol presented in Algorithm~\ref{Alg_Prot}, Alice and Bob use only $mk$ bits of the random string $\boldsymbol\alpha$ whose places are determined according to the set of sub-keys $\mathcal{S}$, and the protocol is implemented by exploiting the simple XOR operation.  
\end{remark}
 
 The random string ${\boldsymbol\alpha}$ is publicly available and the adversary chooses  any recording function $A_1(\boldsymbol\alpha):\{0,1\}^{kn}\rightarrow \{0,1\}^{\beta}$, computes ${\boldsymbol\zeta}=A_1(\boldsymbol\alpha)$, and stores ${\boldsymbol\zeta}$.
 We show that even if the adversary later gets the secret key $\mathbf{Z}$ and his computational power is unbounded, the encryption is secure. It is worth noting that the restriction on the adversary's storage is applied during the transmission of $\boldsymbol\alpha$, after that there is no restriction of the adversary's storage. Without loss of generality, in this paper, results are derived for $\gamma=0.45$.\footnote{Similar results can be derived for any $\gamma<1$.} Next, we present the attack model.
 
 \subsubsection{Attack Model}
 In the bounded storage model, formally, the adversary performs an attack in two phases as follows:
 \begin{itemize}
     \item \textbf{Phase I}: The random string $\boldsymbol\alpha\stackrel{R}{\leftarrow}\{0,1\}^{kn}$ is broadcast. The adversary performs an arbitrary recording function $A_1(\boldsymbol\alpha):\{0,1\}^{kn}\rightarrow \{0,1\}^{\beta}$ on the random string $\boldsymbol\alpha$ and computes ${\boldsymbol\zeta}=A_1(\boldsymbol\alpha)$. At the end of this phase, the adversary stores ${\boldsymbol\zeta}$.
     \item \textbf{Phase II}: The adversary is provided with the ciphertext $\mathbf{C}$, the output of Phase I, ${\boldsymbol\zeta}$, the secret key $\mathbf{Z}$, and infinite computing power and infinite storage space. Using the provided information, the adversary tries to gain information on the message $\mathbf{M}$ by applying any decoding function $A_2({\boldsymbol\zeta},\mathbf{Z},\mathbf{C})$.
 \end{itemize}

We prove that the security of the protocol follows the absolute version  of the notion of semantic security \cite{Goldwasser1984}, i.e., the protocol  is semantically secure in a system allowing a computationally
unbounded adversary, and thus, the protocol is absolutely semantically secure. In other words, we prove that for any recording function $A_1(\boldsymbol\alpha)$ and any decoding algorithm $A_2({\boldsymbol\zeta},\mathbf{Z},\mathbf{C})$, the probability that the adversary with unbounded computational power gains even one bit of information on the message $\mathbf{M}$ is exponentially small in the security parameter $k$. More specifically, consider two equiprobable messages $\mathbf{M}^0$ and $\mathbf{M}^1$. One of the two messages is chosen uniformly at random, encrypted using the provided protocol, and transmitted. The adversary wishes to know which one of the messages is transmitted. We show that using any recording function $A_1(\boldsymbol\alpha)$ and any decoding algorithm $A_2({\boldsymbol\zeta},\mathbf{Z},\mathbf{C})$ with unbounded computational power, the adversary cannot distinguish between $\mathbf{M}^0$ and $\mathbf{M}^1$ from the ciphertext $\mathbf{C}$, except with an exponentially small probability in the security parameter $k$. The following theorem presents the security of the protocol.
\begin{theorem}\label{‌Bounded-s-th}
    For any two equiprobable messages $\mathbf{M}^0$ and $\mathbf{M}^1$ of size $m$, for any recording function $A_1(\boldsymbol\alpha):\{0,1\}^{kn}\rightarrow \{0,1\}^{\beta}$, for any decoding algorithm $A_2$, for $\boldsymbol\alpha\stackrel{R}{\leftarrow}\{0,1\}^{kn}$, and $\mathbf{Z}\stackrel{R}{\leftarrow} \mathbb{Z}_n^k$, the advantage of the adversary in distinguishing  between the encryption of the two messages is upper-bounded  as 
    \begin{align}\nonumber
&\bigg|\mathrm{Pr}\left(A_2({\boldsymbol\zeta},\mathbf{Z},\mathbf{M}^1\oplus \mathbf{X})=1\right)-\mathrm{Pr}\left(A_2({\boldsymbol\zeta},\mathbf{Z},\mathbf{M}^0\oplus \mathbf{X})=1\right)\bigg|\\& \label{theo1ineq}~~~~~~~~~~~~~~~~~~~~~~~~~~~~~<m(2^{-k/6+1}+2^{-0.002kn+2}).
    \end{align}
\end{theorem}
\begin{proof}
    See Section~\ref{Proof bounded Storage Model}.
\end{proof}

\begin{remark}
As long as the shared secret key $\mathbf{Z}$ is not revealed to the adversary, it can be reused to transmit new messages by using new random strings. In other words, the secret key is established between Alice and Bob once and for all messages transmitted.
\end{remark}

\begin{remark}
In practice, $2^{-0.002kn+2}$ is negligibly small compared with $2^{-k/6+1}$, thus, during the informal discussions on the security of the protocol in the bounded storage model, we drop the negligible term $2^{-0.002kn+2}$. For example, for the parameters $n=2^{45}$,
 $m=2^{25}$, and ${k=300}$ the probability of distinguishing between the encryption of the two messages is upper-bounded by $2^{-23}$.
\end{remark}

\section{Proof of Security}\label{Proof bounded Storage Model}
To prove Theorem~\ref{‌Bounded-s-th}, first, we prove bit security (see Proposition~\ref{bounded-lemma-1} for details) demonstrating that if the adversary is given all but the $i$th bit of the message $\mathbf{M}$, the probability of correctly computing the $i$th bit of the final key $\mathbf{X}$ is exponentially small in the security parameter $k$. Subsequently, we establish the connection between bit security and semantic security. More specifically, we show that if the adversary can distinguish between messages $\mathbf{M}^0$ and $\mathbf{M}^1$, then he can compute the $i$th missing bit of the final key, contradicting the bit security (see Section~\ref{sec_proof_theo1} for details).

Next, we provide the formal attack model for the bit security. 
\begin{itemize}
 \item \textbf{Phase I}: The random string $\boldsymbol\alpha\stackrel{R}{\leftarrow}\{0,1\}^{kn}$ is broadcast. The adversary performs an arbitrary recording function $B_1(\boldsymbol\alpha):\{0,1\}^{kn}\rightarrow \{0,1\}^{\beta}$ on the random string $\boldsymbol\alpha$, computes ${{\boldsymbol\eta}=B_1(\boldsymbol\alpha)}$, and stores ${\boldsymbol\eta}$.
     \item \textbf{Phase II}: The adversary is provided with i) all but the $i$th bits of the final key $\mathbf{X}$, denoted as  $\mathbf{X}^{-i}$, ii) the output of Phase I, ${\boldsymbol\eta}$,  iii) the secret key $\mathbf{Z}$, and iv) infinite computing power and infinite storage space. Using the provided information, the adversary tries to compute the $i$th missing bit of the final key $\mathbf{X}$, using any decoding algorithm $B_2({\boldsymbol\eta},\mathbf{Z},\mathbf{X}^{-i})$.
\end{itemize}

Let us show the $i$th bit of the final key $\mathbf{X}$ as $\mathbf{Z}(i,\boldsymbol{\alpha})$, i.e., $\mathbf{Z}(i,\boldsymbol{\alpha})\triangleq X_i$. This notation shows that the $i$th bit of the final key $\mathbf{X}$ is calculated by using the random string $\boldsymbol{\alpha}$ and the secret key $\mathbf{Z}$. Then,
the bit security of the protocol is presented in the following proposition.
\begin{proposition}[Bit security]\label{bounded-lemma-1}
        For any recording function $  B_1(\boldsymbol\alpha):\{0,1\}^{kn}\rightarrow \{0,1\}^{\beta}$, for any decoding algorithm $B_2$, for $\boldsymbol\alpha\stackrel{R}{\leftarrow}\{0,1\}^{kn}$, and $\mathbf{Z}\stackrel{R}{\leftarrow} \mathbb{Z}_n^k$, we have
\begin{align}\label{bit-sec-m-le}
&\bigg|\mathrm{Pr}\left(B_2({\boldsymbol\eta},\mathbf{Z},\mathbf{X}^{-i})=\mathbf{Z}(i,\boldsymbol\alpha)\right)-\dfrac{1}{2}\bigg|<2^{-k/6}+2^{-0.002kn+1}.
\end{align}
\end{proposition}

\subsection{Proof of Proposition~\ref{bounded-lemma-1}}
The main idea behind the proof of bit security is to show that the number of random strings $\underline\alpha\in\{0,1\}^{kn}$ for which any decoding algorithm $B_2$ can provide a desirable result is negligible compared with the number of all random strings $\underline\alpha'$ for which we have $B_1(\underline\alpha')=\underline\eta$. 

\begin{remark}
   It suffices to prove the theorem for the case where the recording function $B_1$ and the decoding algorithm $B_2$ are deterministic \cite{AumannZongRabin02,DziembowskiMaurer02}. This is because a randomized recording function is an algorithm that uses a random help string to compute its output. A randomized algorithm with a fixed help string gives rise to a deterministic algorithm \cite[Remark~3, Page~5]{AumannZongRabin02}. The same argument applies to the decoding algorithm.
\end{remark}

Next, we present some necessary definitions and preliminary results.

\begin{definition}
Let $K=n^k$, $N=2^{nk}$, $(\underline{Z}_1,\ldots,\underline{Z}_K)$  denote an enumeration of all possible secret keys, and $(\underline\alpha_1,\ldots,\underline\alpha_N)$ denote an enumeration of all possible random strings of length $nk$.
\end{definition}

\begin{definition}
For a bit ${w\in\{0,1\}}$, we define $\bar{w}=(-1)^w$, and for a vector $\underline W=(w_1,\ldots,w_t)\in\{0,1\}^t$, we define ${\overline{\underline W}=(\bar{w}_1,\ldots,\bar{w}_t)}$.
\end{definition}

\begin{definition}
For a string $\underline\alpha\in\{0,1\}^{nk}$, we define $\underline\nu(i,\underline\alpha)=(\overline{\underline{Z}_1(i,\underline\alpha)},\ldots,\overline{\underline{Z}_K(i,\underline\alpha)})$.
We use the discrepancy function $d(\underline\nu(i,\underline\alpha))$ to measure the excess of ones over zeros, or vice versa, in the vector $\underline\nu(i,\underline\alpha)$, which is defined as
\begin{align}
   d(\underline\nu(i,\underline\alpha))=\bigg|\sum_{j=1}^{K}\overline{\underline{Z}_j(i,\underline\alpha)}\bigg|.
\end{align}
\end{definition}
\begin{lemma}\label{lemma-bou-dv}
For any $\underline\alpha\in\{0,1\}^{nk}$, such that neither the fraction of ones nor that of zeros in $\underline\alpha$ 
is  no less than $1/8$,\footnote{It is worth noting that $1/8$ is an appropriate number for $\gamma=0.45$,   for other values of $\gamma$ it needs to be changed. However, $1/8$ is not the only suitable number for $\gamma=0.45$; one can use another appropriate value for which the only difference would be the coefficient of the security parameter $k$ on the right-hand side of the inequality in \eqref{theo1ineq}.}
 we have 
\begin{align}
    d(\underline\nu(i,\underline\alpha))\le K2^{-k/3}.
\end{align}
\end{lemma}
\begin{proof}
See Appendix~\ref{app_lem1}.
\end{proof}

 In the next lemma, which is derived from Lemma~\ref{lemma-bou-dv}, we show that for almost all ${\underline\alpha\in\{0,1\}^{nk}}$, we have $d(\underline\nu(i,\underline\alpha))\le K2^{-k/3}$.
\begin{lemma}\label{boundongoodalpha}
    Let $\mathcal{D}$ denote the set of strings $\underline\alpha\in\{0,1\}^{nk}$ for which $d(\underline\nu(i,\underline\alpha))> K2^{-k/3}$, i.e., $\mathcal{D}=\{\underline\alpha\in\{0,1\}^{nk}:d(\underline\nu(i,\underline\alpha))> K2^{-k/3}\}$. Then, an upper-bound on the cardinality of $\mathcal{D}$ is given as $|\mathcal{D}|<2^{0.544kn}$.
\end{lemma}
\begin{proof}
See Appendix~\ref{app_lem2}.
\end{proof}

Next, by using Lemma~\ref{boundongoodalpha}, we present a lemma that is useful in the proof of bit security. 
\begin{lemma}\label{alphagooddelta}
Let $\mathbf{V}$ be the $K\times N$ matrix whose $j$th column is $\underline\nu(i,\underline\alpha_j)
^T,$ where $T$ denotes the transpose operation. Let $\Delta=\mathbf{V}^T\mathbf{V}$, and $\delta_{j,j'}$ denote the $j'$th element of the $j$th row of matrix $\Delta$. Then, for each fixed $j$, the number of elements $\delta_{j,j'}$ in the $j$th row such that $|\delta_{j,j'}|>K2^{
-k/3}$ is at most $2^{0.544kn}$.
\end{lemma}
\begin{proof}
 See Appendix~\ref{app_lem3}.
\end{proof}

In Lemma~\ref{comp-ind-final}, we show that knowing a portion of the final key $\mathbf{X}$ does not provide any information about the missing part.
\begin{lemma}\label{comp-ind-final}
For any $\underline{Z}\in \mathbb{Z}_n^k$, and for any $\boldsymbol\alpha\stackrel{R}{\leftarrow}\{0,1\}^{kn}$, components of $\mathbf{X}\in\{0,1\}^m$, i.e., $X_{1},\ldots,X_{m}$, are statistically independent.
\end{lemma}

\begin{proof} 
See Appendix~\ref{app_lem4}.
\end{proof}

Let ${\mathbf{X}^{-i}=(X_1,\ldots,X_{i-1},X_{i+1},\ldots,X_m)}$ denote all the bits of the final key $\mathbf{X}$ except the $i$th bit. The following lemma shows that knowing $\mathbf{X}^{-i}$ does not provide any information on the secret key $\mathbf{Z}$. More specifically, the mutual information between $\mathbf{X}^{-i}$ and $\mathbf{Z}$ is zero, i.e., $I(\mathbf{Z};\mathbf{X}^{-i})=0$.   
\begin{lemma}\label{xizin}
  For any $\boldsymbol\alpha\stackrel{R}{\leftarrow}\{0,1\}^{kn}$, for any $i\in\{1,\cdots,m\}$, and for any $\mathbf{Z}\stackrel{R}{\leftarrow} \mathbb{Z}_n^k$, the secret key $\mathbf{Z}$ is statistically independent of $\mathbf{X}^{-i}$. More specifically, 
   $\mathrm{Pr}(\mathbf{X}^{-i}|\mathbf{Z})=\mathrm{Pr}(\mathbf{X}^{-i}).$
\end{lemma}
\begin{proof}
See Appendix~\ref{app_lem5}.
\end{proof}

Now, we define a condition under which we say the decoding algorithm $B_2$ using the output of Phase~I is good \cite[Definition~7]{AumannZongRabin02}.
Then, we show that the number of random strings for which $B_2$ is good is very small compared to the number of strings $\underline\alpha\in\{0,1\}^{nk}$ for which $B_1(\underline\alpha)=\underline\eta$.

\begin{definition}[Goodness of the decoding algorithm $B_2$]\label{gooddef} For a string $\underline\alpha\in\{0,1\}^{nk}$, and a fixed ${\underline\eta\in\{0,1\}^\beta}$, derived from Phase I, we say that the decoding algorithm $B_2$ using $\underline\eta$ is good for $\underline\alpha$ if for $\mathbf{Z}\stackrel{R}{\leftarrow} \mathbb{Z}_n^k$, we have 
\begin{align}\label{goodness-f-def}
\bigg|\mathrm{Pr}\left(B_2(\underline\eta,\mathbf{Z},\underline{X}^{-i})=\mathbf{Z}(i,\underline\alpha)\right)-\dfrac{1}{2}\bigg|\ge2^{-k/6}.
\end{align}
\end{definition}
It will turn out that an equivalent inequality as \eqref{goodness-f-def}, which is presented in Lemma~\ref{equi-good-def}, is useful in the proof of bit security. Before presenting Lemma~\ref{equi-good-def}, we define the enumeration of all outputs of the decoding algorithm $B_2$ for all possible secret keys in the following. 
\begin{definition}
For a fixed $\underline{\eta} \in \{0,1\}^\beta$, we define $\underline{H}(i,\underline{\eta}) = (B_2(\underline{\eta},\underline{Z}_1,\underline{X}^{-i}),\ldots,$$B_2(\underline{\eta},\underline{Z}_K,\underline{X}^{-i}))$, i.e., the enumeration of all outputs of the decoding algorithm $B_2$ for all possible secret keys.
\end{definition}
In the following lemma, by using $\underline{H}(i,\underline\eta)$, we provide another form of goodness definition.
\begin{lemma}\label{equi-good-def}
    For a string $\underline\alpha\in\{0,1\}^{nk}$, and a fixed ${\underline\eta\in\{0,1\}^\beta}$, the decoding algorithm $B_2$ using $\underline\eta$ is good for $\underline\alpha$ if 
    \begin{align}
        \left|\overline{\underline{H}(i,\underline\eta)}.\underline\nu(i,\underline\alpha)\right|\ge\dfrac{2K}{2^{k/6}}.
    \end{align}
\end{lemma}
\begin{proof}
See Appendix~\ref{app_lem7}.
\end{proof}

Next, we derive an upper-bound on the number of random strings $\underline\alpha\in\{0,1\}^{nk}$ for which the decoding algorithm $B_2$ using $\underline\eta$ is good. Note that here we consider all possible random strings ${\underline\alpha\in\{0,1\}^{nk}}$, regardless of if $B_1(\underline\alpha)=\underline\eta$ or not.
\begin{lemma}\label{boundongoods}
    Let $L_{\underline{H}(i,\underline\eta)}$ denote the set of all random strings for which $B_2$ using $\underline\eta$ is good, i.e.,
    \begin{align}
        L_{\underline{H}(i,\underline\eta)}=\left\{\underline\alpha\in\{0,1\}^{nk}:\big|\overline{\underline{H}(i,\underline\eta)}.\underline\nu(i,\underline\alpha)\big|\ge\dfrac{2K}{2^{k/6}}\right\}.
    \end{align}
    Then, an upper-bound on the cardinality of $L_{\underline{H}(i,\underline\eta)}$ is given as $|L_{\underline{H}(i,\underline\eta)}|<2^{0.544nk+k/3}$.
\end{lemma}
\begin{proof}
See Appendix~\ref{app_lem8}.  
\end{proof}

The main step of the proof is to show that the number of strings for which $B_2$ using $\underline\eta$ is good, i.e., $|L_{\underline{H}(i,\underline\eta)}|$, is very small compared to the number of strings $\underline\alpha\in\{0,1\}^{nk}$ for which $B_1(\underline\alpha)=\underline\eta,$ i.e., the pre-image of $\underline\eta$ under the recording function $B_1$. The pre-image of $\underline\eta$ is given as
\begin{align}
B_1^{-1}(B_1(\underline\alpha))=\left\{\underline\theta\in\{0,1\}^{nk}:B_1(\underline\theta)=B_1(\underline\alpha)\right\}.
\end{align}
In the following lemma, we show that for all but a tiny fraction of strings $\underline\alpha\in\{0,1\}^{nk}$, the pre-image of $\underline\eta$ under $B_1$, i.e., $B_1^{-1}(B_1(\underline\alpha))$, contains at least $2^{0.548kn}$ strings in $\{0,1\}^{nk}$.
\begin{lemma}\label{preimage}
For any recording function $B_1(\boldsymbol\alpha):\{0,1\}^{kn}\rightarrow \{0,1\}^{\beta}$, for any $\boldsymbol\alpha\stackrel{R}{\leftarrow}\{0,1\}^{kn}$, we have 
\begin{align}
\mathrm{Pr}\left(\left|B_1^{-1}(B_1(\boldsymbol\alpha))\right|< 2^{0.548kn}\right)\le 2^{-0.002kn}.
\end{align}
\end{lemma}
\begin{proof}
See Appendix~\ref{app_lem9}.
\end{proof}
Next, by using Lemmas~\ref{boundongoods} and \ref{preimage}, we show that the probability that the decoding algorithm $B_2$ using $B_1(\boldsymbol\alpha)$ is good for the random string $\boldsymbol\alpha$, is exponentially small in $kn$.

\begin{lemma}\label{probalphaisgood}
    For any recording function $B_1(\boldsymbol\alpha):\{0,1\}^{kn}\rightarrow \{0,1\}^{\beta}$,  any decoding algorithm $B_2$, any $\boldsymbol\alpha\stackrel{R}{\leftarrow}\{0,1\}^{kn}$, we have 
\begin{align}\label{alphinlh}
\mathrm{Pr}\left(\left|\overline{{\mathbf{H}}(i,{\boldsymbol\eta})}.{\boldsymbol\nu}(i,\boldsymbol\alpha)\right|\ge\dfrac{2K}{2^{k/6}}\right)\le 2^{-0.002kn+1}.
\end{align}
\end{lemma}
\begin{proof}
See Appendix~\ref{app_lem10}.
\end{proof}
Now, we are ready to prove Proposition~\ref{bounded-lemma-1}, i.e., bit security. The probability that the adversary can compute the missing bit $X_i$ by using the output of Phase I, i.e., ${\boldsymbol{\eta}}$, the secret key $\mathbf{Z}$, and all other bits of the final key  $\mathbf{X}^{-i}$, is given as
\begin{align}\nonumber
 &\mathrm{Pr}\left(B_2({\boldsymbol{\eta}},\mathbf{Z},\mathbf{X}^{-i})=\mathbf{Z}(i,\boldsymbol\alpha)\right)=\\&\nonumber\mathrm{Pr}\left(B_2({\boldsymbol{\eta}},\mathbf{Z},\mathbf{X}^{-i})=\mathbf{Z}(i,\boldsymbol\alpha)\bigg|\boldsymbol\alpha\in L_{{\mathbf{H}}(i,{\boldsymbol{\eta}})}\right)\mathrm{Pr}\left(\boldsymbol\alpha\in L_{{\mathbf{H}}(i,{\boldsymbol{\eta}})}\right)+\\&\label{p-cal-mis}
 \mathrm{Pr}\left(B_2({\boldsymbol{\eta}},\mathbf{Z},\mathbf{X}^{-i})=\mathbf{Z}(i,\boldsymbol\alpha)\bigg|\boldsymbol\alpha\notin L_{{\mathbf{H}}(i,{\boldsymbol{\eta}})}\right)\mathrm{Pr}\left(\boldsymbol\alpha\notin L_{{\mathbf{H}}(i,{\boldsymbol{\eta}})}\right).
\end{align}
 Next, by using \eqref{p-cal-mis}, we derive an upper-bound and a lower-bound on the probability of computing the missing bit $X_i$, $\mathrm{Pr}\left(B_2({\boldsymbol{\eta}},\mathbf{Z},\mathbf{X}^{-i})=\mathbf{Z}(i,\boldsymbol\alpha)\right)$.  
 
 The upper-bound is given as
\begin{align}\nonumber
 &\mathrm{Pr}\left(B_2({\boldsymbol{\eta}},\mathbf{Z},\mathbf{X}^{-i})=\mathbf{Z}(i,\boldsymbol\alpha)\right)\stackrel{(a)}<\\&\nonumber
 \mathrm{Pr}\left(B_2({\boldsymbol{\eta}},\mathbf{Z},\mathbf{X}^{-i})=\mathbf{Z}(i,\boldsymbol\alpha)\bigg|\boldsymbol\alpha\notin L_{{\mathbf{H}}(i,{\boldsymbol{\eta}})}\right)+\mathrm{Pr}\left(\boldsymbol\alpha\in L_{{\mathbf{H}}(i,{\boldsymbol{\eta}})}\right)\stackrel{(b)}<
 \\&\label{p-cal-mis-upper}\dfrac{1}{2}+{2^{-k/6}}+2^{-0.002kn+1},
\end{align}
 where $(a)$ follows from ignoring the probabilities $
 \mathrm{Pr}\left(B_2({\boldsymbol{\eta}},\mathbf{Z},\mathbf{X}^{-i})=\mathbf{Z}(i,\boldsymbol\alpha)\big|\boldsymbol\alpha\in L_{{\mathbf{H}}(i,{\boldsymbol{\eta}})}\right)$ and $\mathrm{Pr}\left(\boldsymbol\alpha\notin L_{{\mathbf{H}}(i,{\boldsymbol{\eta}})}\right)$ from the right-hand side of \eqref{p-cal-mis}, and $(b)$ follows because i) according to  Lemma~\ref{probalphaisgood}, an upper-bound on $\mathrm{Pr}\left(\boldsymbol\alpha\in L_{{\mathbf{H}}(i,{\boldsymbol{\eta}})}\right)$ is given as ${\mathrm{Pr}\left(\boldsymbol\alpha\in L_{{\mathbf{H}}(i,{\boldsymbol{\eta}})}\right)\le 2^{-0.002kn+1}}$ and ii) from the definition of goodness of algorithm  $B_2$, Def.~\ref{gooddef}, when $\boldsymbol\alpha\notin L_{{\mathbf{H}}(i,{\boldsymbol{\eta}})}$, we have $\left|\mathrm{Pr}\left(B_2({\boldsymbol{\eta}},\mathbf{Z},\underline{X}^{-i})=\mathbf{Z}(i,\boldsymbol\alpha)\right)-\dfrac{1}{2}\right|< 2^{-k/6}$, which, in turn, implies that $$\mathrm{Pr}\left(B_2({\boldsymbol{\eta}},\mathbf{Z},\mathbf{X}^{-i})=\mathbf{Z}(i,\boldsymbol\alpha)\big|\boldsymbol\alpha\notin L_{{\mathbf{H}}(i,{\boldsymbol{\eta}})}\right)<\dfrac{1}{2}+{2^{-k/6}}.$$

 The lower-bound is given as
 \begin{align}\nonumber
 &\mathrm{Pr}\left(B_2({\boldsymbol{\eta}},\mathbf{Z},\mathbf{X}^{-i})=\mathbf{Z}(i,\boldsymbol\alpha)\right)\stackrel{(a)}>\\&\nonumber
 \mathrm{Pr}\left(B_2({\boldsymbol{\eta}},\mathbf{Z},\mathbf{X}^{-i})=\mathbf{Z}(i,\boldsymbol\alpha)\bigg|\boldsymbol\alpha\notin L_{{\mathbf{H}}(i,{\boldsymbol{\eta}})}\right)\mathrm{Pr}\left(\boldsymbol\alpha\notin L_{{\mathbf{H}}(i,{\boldsymbol{\eta}})}\right)\stackrel{(b)}>\\&\label{p-cal-mis-lower}
 \dfrac{1}{2}-{2^{-k/6}}-2^{-0.002kn+1},
\end{align}
where $(a)$ follows from ignoring the probabilities $
 \mathrm{Pr}\left(B_2({\boldsymbol{\eta}},\mathbf{Z},\mathbf{X}^{-i})=\mathbf{Z}(i,\boldsymbol\alpha)\big|\boldsymbol\alpha\in L_{{\mathbf{H}}(i,{\boldsymbol{\eta}})}\right)$ and $\mathrm{Pr}\left(\boldsymbol\alpha\in L_{{\mathbf{H}}(i,{\boldsymbol{\eta}})}\right)$ from the right-hand side of \eqref{p-cal-mis}, and $(b)$ follows because i) according to Lemma~\ref{probalphaisgood}, a lower-bound on ${\mathrm{Pr}\left(\boldsymbol\alpha\notin L_{{\mathbf{H}}(i,{\boldsymbol{\eta}})}\right)}$ is given as ${\mathrm{Pr}\left(\boldsymbol\alpha\notin L_{{\mathbf{H}}(i,{\boldsymbol{\eta}})}\right)\ge 1-2^{-0.002kn+1}}$ and ii) from the definition of goodness of algorithm  $B_2$, when $\boldsymbol\alpha\notin L_{{\mathbf{H}}(i,{\boldsymbol{\eta}})}$, we have $\left|\mathrm{Pr}\left(B_2({\boldsymbol{\eta}},\mathbf{Z},\underline{X}^{-i})=\mathbf{Z}(i,\boldsymbol\alpha)\right)-\dfrac{1}{2}\right|< 2^{-k/6}$, which, in turn, implies that $$\mathrm{Pr}\left(B_2({\boldsymbol{\eta}},\mathbf{Z},\mathbf{X}^{-i})=\mathbf{Z}(i,\boldsymbol\alpha)\big|\boldsymbol\alpha\notin L_{{\mathbf{H}}(i,{\boldsymbol{\eta}})}\right)>\dfrac{1}{2}-{2^{-k/6}}.$$

 Finally, using the bounds in \eqref{p-cal-mis-upper} and \eqref{p-cal-mis-lower}, we have 
\begin{align}\nonumber
&\bigg|\mathrm{Pr}\left(B_2({\boldsymbol{\eta}},\mathbf{Z},\mathbf{X}^{-i})=\mathbf{Z}(i,\boldsymbol\alpha)\right)-\dfrac{1}{2}\bigg|<2^{-k/6}+2^{-0.002kn+1}.
 \end{align}
 which completes the proof of Proposition~\ref{bounded-lemma-1}.
\subsection{Proof of Theorem~\ref{‌Bounded-s-th}}\label{sec_proof_theo1}
Having proved the bit security, we show the relationship between the bit security and semantic security in the next lemma. More specifically, we show that if the adversary can break the semantic security of the protocol, then he can compute the $i$th missing bit in Phase II, which contradicts the bit security.
\begin{lemma}
    For any two equiprobable messages $\mathbf{M}^0$ and $\mathbf{M}^1$ of size $m$, any recording function $A_1(\boldsymbol\alpha):\{0,1\}^{kn}\rightarrow \{0,1\}^{\beta}$, any decoding algorithm $A_2$, for $\boldsymbol\alpha\stackrel{R}{\leftarrow}\{0,1\}^{kn}$, and $\mathbf{Z}\stackrel{R}{\leftarrow} \mathbb{Z}_n^k$, if
    \begin{align}\label{bitsec-r-sem1}
&\bigg|\mathrm{Pr}\left(A_2({\boldsymbol\zeta},\mathbf{Z},\mathbf{M}^1\oplus \mathbf{X})=1\right)-\mathrm{Pr}\left(A_2({\boldsymbol\zeta},\mathbf{Z},\mathbf{M}^0\oplus \mathbf{X})=1\right)\bigg|=\epsilon,
    \end{align}
then, there is an $i$, a recording function $B_1$, and a decoding algorithm $B_2$ such that 
\begin{align}\label{bitsec-r-sem2}
\left|\mathrm{Pr}\left(B_2(B_1(\boldsymbol\alpha),\mathbf{Z},\mathbf{X}^{-i})=\mathbf{Z}(i,\boldsymbol\alpha)\right)-\dfrac{1}{2}\right|\ge \dfrac{\epsilon}{2m}.
\end{align}
\end{lemma}
\begin{proof}
The proof follows similar steps as for the proof of Lemma~23 in \cite{AumannZongRabin02}.
\end{proof}
Finally, combining \eqref{bit-sec-m-le}, \eqref{bitsec-r-sem1}, and \eqref{bitsec-r-sem2} we have 
\begin{align}\nonumber
\nonumber
&\bigg|\mathrm{Pr}\left(A_2({\boldsymbol\zeta},\mathbf{Z},\mathbf{M}^1\oplus \mathbf{X})=1\right)-\mathrm{Pr}\left(A_2({\boldsymbol\zeta},\mathbf{Z},\mathbf{M}^0\oplus \mathbf{X})=1\right)\bigg|\\& ~~~~~~~~~~~~~~~~~~~~~~~~~~~<m(2^{-k/6+1}+2^{-0.002kn+2}),
\end{align}
 which completes the proof of Theorem~\ref{‌Bounded-s-th}.

\section{Conclusions}\label{Conclusions}
In this paper, we considered a general bounded storage model where the adversary can access all bits of the random string, and store the output of any Boolean function on the string. We reaffirm that the protocol provided by Maurer in \cite{Maurer} is absolutely semantically secure in the general bounded storage model with a secret key and a random string of efficient sizes.

The interesting future work would include leveraging the concept of the bounded storage model to develop absolutely semantically secure protocols for emerging applications in communication systems. Examples include but are not limited to cloud storage, cloud computing, and secure multi-party communication in the presence of an adversary with unbounded storage and unbounded computational power.

\section{Appendix}
\subsection{Proof of Lemma~\ref{lemma-bou-dv}}\label{app_lem1}
Let $p$ denote the fraction of ones, and $p'=1-p$ denote the fraction of zeros in $\underline\alpha$, then we have
\begin{align}\nonumber d(\underline\nu(i,\underline\alpha))&\stackrel{(a)}=K\bigg|\mathrm{Pr}(\mathbf{Z}(i,\underline{\alpha})=0)-\mathrm{Pr}(\mathbf{Z}(i,\underline{\alpha})=1)\bigg|\\&\nonumber
\stackrel{(b)}=K|p'-p|^k\\&\stackrel{(c)}\le K2^{-k/3},
\end{align}
where equality $(a)$ follows from the fact that  i) $\mathrm{Pr}(\mathbf{Z} (i,\underline{\alpha})=0)$ ($\mathrm{Pr}(\mathbf{Z} (i,\underline{\alpha})=1)$) is calculated as the ratio of the number of zeros (ones) in $\underline\nu(i,\underline\alpha)$ over the number all elements in $\underline\nu(i,\underline\alpha)$ which is $K$, and ii) 
$d(\underline\nu(i,\underline\alpha))$ measures the excess of ones over zeros, or vice versa, in the string $\underline\nu(i,\underline\alpha)$,
and equality $(b)$ follows because for $\mathbf{Z}\stackrel{R}{\leftarrow} \mathbb{Z}_n^k$, we have
    \begin{align}\nonumber
    &\mathrm{Pr}(\mathbf{Z}(i,\underline{\alpha})=1)=\sum_{j \text{odd}}{\binom{k}{j}}p^j{p'}^{k-j},\\&
    \mathrm{Pr}(\mathbf{Z}(i,\underline{\alpha})=0)=\sum_{j \text{even}}{\binom{k}{j}}p^j{p'}^{k-j},
    \end{align}
    thus, using the Binomial theorem \cite[Page~162]{BGraham1994Concrete}, we have 
    \begin{align}
        &\left|\mathrm{Pr}(\mathbf{Z}(i,\underline{\alpha})=0)-\mathrm{Pr}(\mathbf{Z}(i,\underline{\alpha})=1)\right|\\&\nonumber
       =\left|\sum_{j~ \text{even}}{\binom{k}{j}}p^j{p'}^{k-j}-\sum_{j~ \text{odd}}{\binom{k}{j}}p^j{p'}^{k-j}\right|\\&\nonumber
        =|p'-p|^k.
    \end{align}
    Finally,  inequality $(c)$ follows because the fractions of ones and zeros in $\underline\alpha$ are both no less than $1/8$, thus,
    \begin{align}\nonumber
       |p-p'|^k&\le(\dfrac{7}{8}-\dfrac{1}{8})^k\\&<2^{-k/3}.
    \end{align}

\subsection{Proof of Lemma~\ref{boundongoodalpha}}\label{app_lem2}
Let $o(\underline\alpha)$ denote the number of ones in the string $\underline\alpha$. Then,  according to Lemma~1, the string $\underline\alpha$  belongs to $\mathcal{D}$, i.e., $\underline\alpha\in\mathcal{D}$, implies that 
  $o(\underline\alpha)<\dfrac{nk}{8}$ or $o(\underline\alpha)>\dfrac{7nk}{8}$. Thus, for a random string $\boldsymbol\alpha$, the probability of the event  $\boldsymbol\alpha\in\mathcal{D}$ is upper-bounded as 
  \begin{align}\label{alphagoodpro}
\mathrm{Pr}\left(\boldsymbol\alpha\in\mathcal{D}\right)\le\mathrm{Pr}\left(o(\boldsymbol\alpha)<\dfrac{nk}{8}~\text{or}~o(\boldsymbol\alpha)>\dfrac{7nk}{8}\right).
  \end{align}
  Using Stirling's approximation \cite[Page~598]{BGraham1994Concrete}, for $\boldsymbol\alpha\stackrel{R}{\leftarrow}\{0,1\}^{kn},$ we have 
  \begin{align}\nonumber
\mathrm{Pr}\left(o(\boldsymbol\alpha)<\dfrac{nk}{8}~\text{or}~o(\boldsymbol\alpha)>\dfrac{7nk}{8}\right)&\le 2^{-nk(1-h(1/8))+1}\\&\nonumber
=2^{-0.4563kn+1}\\&\nonumber=2^{-0.4560kn-0.0003kn+1}\\&\label{bound-on-PD}\stackrel{(a)}< 2^{-0.456kn}
  \end{align}
  where $h(.)$ is the binary entropy, i.e., $h(1/8)=1/8\log(8)+7/8\log(8/7)$, and $(a)$ follows because in practice, $n$ is very large\footnote{In practice, $n$ is greater than $2^{45}$.} such that $-0.0003kn+1 <0$.
  %
  %
  Finally, using \eqref{bound-on-PD}, we have 
\begin{align}\nonumber
|\mathcal{D}|&=2^{kn}\mathrm{Pr}\left(\boldsymbol\alpha\in\mathcal{D}\right)\\&\nonumber< 2^{0.544kn}.   
\end{align}  

\subsection{Proof of Lemma~\ref{alphagooddelta}}\label{app_lem3}
 From the definition of $\Delta$, the value of $|\delta_{j,j'}|$ is given as
  \begin{align} \nonumber  &|\delta_{j,j'}|=|\underline\nu(i,\underline\alpha_j).\underline\nu(i,\underline\alpha_{j'})|=\\&\nonumber
  \left|(\overline{\underline{Z}_1(i,\underline\alpha_j)},\ldots,\overline{\underline{Z}_K(i,\underline\alpha_{j})}).(\overline{\underline{Z}_1(i,\underline\alpha_{j'})},\ldots,\overline{\underline{Z}_K(i,\underline\alpha_{j'})})\right|=
  \\&\nonumber 
\left|\overline{\underline{Z}_1(i,\underline\alpha_j)}~\overline{\underline{Z}_1(i,\underline\alpha_{j'})}+\cdots+\overline{\underline{Z}_K(i,\underline\alpha_j)}~\overline{\underline{Z}_K(i,\underline\alpha_{j'})}\right|
  \stackrel{(a)}=\\& d(\underline\nu(\underline\alpha_j\oplus\underline\alpha_{j'})),
  \end{align}
where $.$ represents the inner product of two vectors and $(a)$ follows from the fact that $\overline{\underline{Z}_{j''}(i,\underline\alpha_j)}~\overline{\underline{Z}_{j''}(i,\underline\alpha_{j'})}$ can be written as
\begin{align}\nonumber
\overline{\underline{Z}_{j''}(i,\underline\alpha_j)}~\overline{\underline{Z}_{j''}(i,\underline\alpha_{j'})}&=(-1)^{\underline{Z}_{j''}(i,\underline\alpha_j)} (-1)^{\underline{Z}_{j''}(i,\underline\alpha_{j'})}\\&\nonumber
=(-1)^{\underline{Z}_{j''}(i,\underline\alpha_j)\oplus\underline{Z}_{j''}(i,\underline\alpha_{j'})}\\&\nonumber
\stackrel{(a)}=(-1)^{\underline{Z}_{j''}(i,\underline\alpha_j\oplus\alpha_{j'})}
\\&
=\overline{\underline{Z}_{j''}(i,\underline\alpha_j\oplus\alpha_{j'})},
\end{align}
  where $(a)$ follows from the definition of $\underline{Z}_{j''}(i,\underline\alpha_j)$  in \eqref{eq2ql}. Thus, the $j$th row of matrix $\Delta$ is $$\left(d(\underline\nu(\underline\alpha_j\oplus\underline\alpha_{1}), d(\underline\nu(\underline\alpha_j\oplus\underline\alpha_{2}), \ldots,d(\underline\nu(\underline\alpha_j\oplus\underline\alpha_{N})\right).$$
  Since the sequence of strings $\underline\alpha_j\oplus\underline\alpha_{1},\underline\alpha_j\oplus\underline\alpha_{2},\ldots,\underline\alpha_j\oplus\underline\alpha_{N}$ enumerates all possible binary strings with length $kn$, it follows directly from Lemma~\ref{boundongoodalpha} that the number of elements $\delta_{j,j'}$ in the $j$th row of $\Delta$ such that $|\delta_{j,j'}|>K2^{-k/3}$ is at most $2^{0.544kn}$.

\subsection{Proof of Lemma~\ref{comp-ind-final}}\label{app_lem4}
 Recall that $X_i=\bigoplus_{j=1}^k\boldsymbol{\alpha}^{(j)}[Z_{j}\pmb{+}i-1]$.
We show that for any two distinct $i$ and $i'$, the sets of bits of random string $\boldsymbol\alpha$ that are used to compute  $X_{i}$ and $X_{i'}$ have no bit in common, thus,  the components of $\mathbf{X}$ are computed by using different bits of the random string $\boldsymbol\alpha$ and consequently, they are statistically independent.
By looking at the procedure of computing $X_{i}$ and $X_{i'}$, we can see that the sets of bits of random string $\boldsymbol\alpha$ that are used to compute  $X_{i}$ and $X_{i'}$ have no bit in common if the components of the $k$-tuple $\mathbf{D}_{ii'}=\mathbf{S}^{(i)}-\mathbf{S}^{(i')}$ are all non-zero. Using \eqref{eq1},  we have $\mathbf{D}_{ii'}=(i-i')\bm{1}$, which is a $k$-tuple with non-zero components for any two distinct $i$ and $i'$.

\subsection{Proof of Lemma~\ref{xizin}}\label{app_lem5}
Let $(z_1,\ldots,z_k)$ denote an arbitrary realization of the secret key $\mathbf{Z}$,  $p_1$ denote the probability that any bit in the random string $\boldsymbol\alpha$ is one, and $p_1'=1-p_1$ denote the probability that any bit in the random string $\boldsymbol\alpha$ is zero. Then, the probability $\mathrm{Pr}(\mathbf{X}^{-i}|\mathbf{Z})$ is calculated as 
  \begin{align}\nonumber
 \mathrm{Pr}(\mathbf{X}^{-i}|\mathbf{Z})&=\mathrm{Pr}\bigg(\mathbf{X}^{-i}=(x_1,\ldots,x_{i-1},x_{i+1},\ldots,x_m)\big|\mathbf{Z}=(z_1,\ldots,z_k)\bigg)\\&\nonumber
 =\mathrm{Pr}\bigg(\bigoplus_{j=1}^k\boldsymbol{\alpha}^{(j)}[S^{(1)}_{j}]=x_1,\ldots,\bigoplus_{j=1}^k\boldsymbol{\alpha}^{(j)}[S^{(m)}_j]=x_m)\big|\mathbf{Z}=(z_1,\ldots,z_k)\bigg)\\&\nonumber\stackrel{(a)}=\mathrm{Pr}\bigg(\bigoplus_{j=1}^k\boldsymbol{\alpha}^{(j)}[z_{j}]=x_1,\ldots,\bigoplus_{j=1}^k\boldsymbol{\alpha}^{(j)}[z_{j}+m-1]=x_m\bigg) \\&\nonumber \stackrel{(b)}=\prod_{l\in\{1,\cdots,m\}\setminus\{i\}}\mathrm{Pr}\left(\bigoplus_{j=1}^k\boldsymbol{\alpha}^{(j)}[z_{j}+l-1]=x_l\right)
 \\&\nonumber
 =\prod_{l\in\{1,\cdots,m\}\setminus\{i\}}\bigg( \sum_{t=0}^{\lfloor k/2 \rfloor} {\binom{k}{2t+x_l}} p_1^{2t+x_l}{p_1'}^{k-(2t+x_l)}\bigg)\\&\nonumber
 =\sum_{\underline{{Z}}\in\{0,\cdots,n-1\}^k}\bigg(\prod_{l\in\{1,\cdots,m\}\setminus\{i\}}\bigg( \sum_{t=0}^{\lfloor k/2 \rfloor} {\binom{k}{2t+x_l}} p_1^{2t+x_l}{p_1'}^{k-(2t+x_l)}\bigg)\bigg)\dfrac{1}{n^k}
 \\&\nonumber
   \stackrel{(c)}=\sum_{\underline{{Z}}\in\{0,\cdots,n-1\}^k}\bigg(\prod_{l\in\{1,\cdots,m\}\setminus\{i\}}\bigg(\sum_{t=0}^{\lfloor k/2 \rfloor} {\binom{k}{2t+x_l}} p_1^{2t+x_l}{p_1'}^{k-(2t+x_l)}\bigg)\bigg)\mathrm{Pr}(\mathbf{Z}=\underline{Z})
  \\&\nonumber
 =\sum_{\underline{{Z}}\in\{0,\cdots,n-1\}^k}\left( 
 \mathrm{Pr}(\mathbf{X}^{-i}|\mathbf{Z})\right)\mathrm{Pr}(\mathbf{Z}=\underline{Z})\\&
 =
 \mathrm{Pr}(\mathbf{X}^{-i}),
  \end{align}
  where $\lfloor.\rfloor$ represents the floor function, $(a)$ follows from \eqref{eq1}, $(b)$ follows from Lemma~\eqref{comp-ind-final}, i.e., for  $l\ne l'$ the set of bits of $\boldsymbol\alpha$ that are used  to compute $x_l$ has no common bit with the set of bits of $\boldsymbol\alpha$ that are used  to compute $x_{l'}$, and $(c)$ follows because  $\mathbf{Z}$ is chosen uniformly at random from $\mathbb{Z}_n^k$, i.e., $\mathbf{Z}\stackrel{R}{\leftarrow} \mathbb{Z}_n^k$.

\subsection{Proof of Lemma~\ref{equi-good-def}}\label{app_lem7}
First, let us suppose that $\overline{\underline{H}(i,\underline\eta)}.\underline\nu(i,\underline\alpha)\ge \dfrac{2K}{2^{k/6}}$. Then, we have 
\begin{align}\nonumber
&\mathrm{Pr}\left(B_2(\underline\eta,\mathbf{Z},\underline{X}^{-i})=\mathbf{Z}(i,\underline\alpha)\right)\\&\nonumber\stackrel{(a)}=\dfrac{\overline{\underline{H}(i,\underline\eta)}.\underline\nu(i,\underline\alpha)+\dfrac{K-\overline{\underline{H}(i,\underline\eta)}.\underline\nu(i,\underline\alpha)}{2}}{K}\\&
=\dfrac{1}{2}+\dfrac{\overline{\underline{H}(i,\underline\eta)}.\underline\nu(i,\underline\alpha)}{2K},
\end{align}
where $(a)$ follows because $\overline{\underline{H}(i,\underline\eta)}.\underline\nu(i,\underline\alpha)$ is equivalent to the difference between the number of zero and non-zero elements in the tuple $\overline{\underline{H}(i,\underline\eta)}-\underline\nu(i,\underline\alpha)$. 
By using $\overline{\underline{H}(i,\underline\eta)}.\underline\nu(i,\underline\alpha)\ge \dfrac{2K}{2^{k/6}}$, we have 
\begin{align}\label{samedef-u}
\mathrm{Pr}\left(B_2(\underline\eta,\mathbf{Z},\underline{X}^{-i})=\mathbf{Z}(i,\underline\alpha)\right)-\dfrac{1}{2}\ge\dfrac{1}{2^{k/6}}.
\end{align}
Now, suppose that $\overline{\underline{H}(i,\underline\eta)}.\underline\nu(i,\underline\alpha)\le -\dfrac{2K}{2^{k/6}}$. Then, we have 
\begin{align}\nonumber
\mathrm{Pr}\big(B_2(\underline\eta,\mathbf{Z},\underline{X}^{-i})&=\mathbf{Z}(i,\underline\alpha)\big)\\&\nonumber=\dfrac{\dfrac{K+\overline{\underline{H}(i,\underline\eta)}.\underline\nu(i,\underline\alpha)}{2}}{K}\\&
=\dfrac{1}{2}+\dfrac{\overline{\underline{H}(i,\underline\eta)}.\underline\nu(i,\underline\alpha)}{2K}.
\end{align}
By using $\overline{\underline{H}(i,\underline\eta)}.\underline\nu(i,\underline\alpha)\le -\dfrac{2K}{2^{k/6}}$, we have 
\begin{align}\label{samedef-l}
\mathrm{Pr}\left(B_2(\underline\eta,\mathbf{Z},\underline{X}^{-i})=\mathbf{Z}(i,\underline\alpha)\right)-\dfrac{1}{2}\le-\dfrac{1}{2^{k/6}}.
\end{align}
Finally, combining \eqref{samedef-u} and \eqref{samedef-l}, we have $\left|\mathrm{Pr}\left(B_2(\underline\eta,\mathbf{Z},\underline{X}^{-i})=\mathbf{Z}(i,\underline\alpha)\right)-\dfrac{1}{2}\right|\ge\dfrac{1}{2^{k/6}}$, which completes the proof.

\subsection{Proof of Lemma~\ref{boundongoods}}\label{app_lem8}
Let $$L^+_{\underline{H}(i,\underline\eta)}=\left\{\underline\alpha\in\{0,1\}^{nk}:\overline{\underline{H}(i,\underline\eta)}.\underline\nu(i,\underline\alpha)\ge\dfrac{2K}{2^{k/6}}\right\},$$
$$L^-_{\underline{H}(i,\underline\eta)}=L_{\underline{H}(i,\underline\eta)}-L^+_{\underline{H}(i,\underline\eta)},$$ and $\underline\xi$ be a binary row vector of length $N=2^{nk}$, where $\underline\xi_j=1$ indicates that $\underline\alpha_j\in L^+_{\underline{H}(i,\underline\eta)}$; otherwise, $\underline\xi_j=0$. Then, we have
\begin{align}\nonumber
\dfrac{2K}{2^{k/6}}\left|L^+_{\underline{H}(i,\underline\eta)}\right|&\stackrel{(a)}\le \overline{\underline{H}(i,\underline\eta)}.\mathbf{V}.\underline\xi\\&\nonumber
\stackrel{(b)}\le\|\overline{\underline{H}(i,\underline\eta)}\|\|\mathbf{V}.\underline\xi\|\\&\label{p-b-g-1}
\stackrel{(c)}\le\sqrt{K}\|\mathbf{V}.\underline\xi\|,
\end{align}
where $\|.\|$ represents the Euclidean norm,  $(a)$ follows from the definition of $L^+_{\underline{H}(i,\underline\eta)}$, $(b)$ follows from the Cauchy-Schwartz inequality, and $(c)$ comes from the fact that the Euclidean norm of a binary vector with length $K$ is calculated as $\sqrt{K}$. The remaining task is to derive an upper-bound on $\|\mathbf{V}.\underline\xi\|$. To this end, first, we find an upper-bound on $\|\mathbf{V}.\underline\xi\|^2,$ which is given as
\begin{align}\nonumber
\|\mathbf{V}.\underline\xi\|^2&= \underline\xi^T.\mathbf{V}^T \mathbf{V}.\underline\xi\\&\nonumber
=\sum_{j=1}^{N}\sum_{j'=1}^{N}\delta_{jj'}\underline\xi_j\underline\xi_{j'}\\&\nonumber
\stackrel{(a)}\le \sum_{j=1}^{N}\sum_{j'=1}^{N}|\delta_{jj'}|\underline\xi_j\underline\xi_{j'}\\&\nonumber
=\sum_{j=1}^{N}\underline\xi_j\bigg(\sum_{j':|\delta_{jj'}|>K2^{-k/3}}|\delta_{jj'}|\underline\xi_{j'}+\sum_{j':|\delta_{jj'}|\le K2^{-k/3}}|\delta_{jj'}|\underline\xi_{j'}\bigg)\\&\nonumber
\stackrel{(b)}\le \sum_{j=1}^{N}\underline\xi_j\bigg(\sum_{j':|\delta_{jj'}|>K2^{-k/3}}|\delta_{jj'}|+\sum_{j':|\delta_{jj'}|\le K2^{-k/3}}|\delta_{jj'}|\underline\xi_{j'}\bigg)\\&\nonumber
\stackrel{(c)}\le \left|L^+_{\underline{H}(i,\underline\eta)}\right|\left(K2^{0.544kn}+\sum_{j':|\delta_{jj'}|\le K2^{-k/3}}|\delta_{jj'}|\underline\xi_{j'}\right)\\&\label{p-b-g-2}
\stackrel{(d)}\le \left|L^+_{\underline{H}(i,\underline\eta)}\right|\left(K2^{0.544kn}+\left|L^+_{\underline{H}(i,\underline\eta)}\right|K2^{-k/3}\right),
\end{align}
where $(a)$ follows because $|\delta_{jj'}|\ge\delta_{jj'}$, $(b)$ follows because for each $j$, we take the summation over all $|\delta_{jj'}|>K2^{-k/3}$, regardless of if $\underline\xi_{j'}=1$ or not, $(c)$ follow from Lemma~\ref{alphagooddelta} and the fact that the maximum value of $|\delta_{jj'}|$ is $K$ (recall that $|\delta_{jj'}|=|\underline\nu(i,\underline\alpha_j).\underline\nu(i,\underline\alpha_{j'})|$, where $\underline\nu(i,\underline\alpha_j)$ and $\underline\nu(i,\underline\alpha_{j'})$ are binary vectors of length $K$), and $(d)$ follows because for ${j':|\delta_{jj'}|\le K2^{-k/3}}$, we have $\sum_{j':|\delta_{jj'}|\le K2^{-k/3}}|\delta_{jj'}|\underline\xi_{j'}\le \left|L^+_{\underline{H}(i,\underline\eta)}\right|K2^{-k/3}$.

Now, combining \eqref{p-b-g-1} and the bound in \eqref{p-b-g-2}, we have 
\begin{align}\nonumber
    &\dfrac{2K}{2^{k/6}}\left|L^+_{\underline{H}(i,\underline\eta)}\right|\le \sqrt{K}\left|L^+_{\underline{H}(i,\underline\eta)}\right|^{1/2}\big(K2^{0.544kn}+\left|L^+_{\underline{H}(i,\underline\eta)}\right|K2^{-k/3}\big)^{1/2}.
\end{align}
Some algebra shows that
\begin{align}\label{plusbounf}
\left|L^+_{\underline{H}(i,\underline\eta)}\right|\le \dfrac{2^{0.544kn+k/3}}{3}.
\end{align}
Following the same approach used to derive the upper-bound \eqref{plusbounf}, an upper-bound on $\left|L^-_{\underline{H}(i,\underline\eta)}\right|$ is given as 
\begin{align}\label{nplusbounf}
\left|L^-_{\underline{H}(i,\underline\eta)}\right|\le \dfrac{2^{0.544kn+k/3}}{3}.
\end{align}
Finally, using \eqref{plusbounf} and \eqref{nplusbounf}, an upper-bound on $\left|L_{\underline{H}(i,\underline\eta)}\right|$ is given as
\begin{align}\nonumber
\left|L_{\underline{H}(i,\underline\eta)}\right|&=\left|L^+_{\underline{H}(i,\underline\eta)}\right|+\left|L^-_{\underline{H}(i,\underline\eta)}\right| \\&\nonumber \le \dfrac{2}{3} 2^{0.544kn+{k/3}}\\&< 2^{0.544kn+{k/3}}.
\end{align}

\subsection{Proof of Lemma~\ref{preimage}}\label{app_lem9}
The recording algorithm $B_1$ maps the set of all binary strings $\{0,1\}^{kn}$ into $2^\beta$ disjoint subsets $\mathcal{F}_1,\ldots,\mathcal{F}_{2^\beta}$, where $\beta=0.45kn$. Thus, we have 
\begin{align}\nonumber
&\mathrm{Pr}\left(\left|B_1^{-1}(B_1(\boldsymbol\alpha))\right|< 2^{0.548kn}\right)\\&\nonumber=\dfrac{\{\underline\alpha: \left|B_1^{-1}(B_1(\underline\alpha))\right|< 2^{0.548kn}\}}{2^{kn}}\\&\nonumber=\dfrac{\sum_{j:|\mathcal{F}_j|<2^{0.548kn}}|\mathcal{F}_j|}{2^{kn}}\\&\nonumber\le\dfrac{2^{0.548kn}2^{0.45kn}}{2^{kn}}\\&=2^{-0.002kn}.
\end{align}

\subsection{Proof of Lemma~\ref{probalphaisgood}}\label{app_lem10}
By using the law of total probability, we have 
\begin{align}\nonumber
&\mathrm{Pr}\left(\left|\overline{{\mathbf{H}}(i,{\boldsymbol\eta})}.{\boldsymbol\nu}(i,\boldsymbol\alpha)\right|\ge\dfrac{2K}{2^{k/6}}\right) \\&\nonumber=\mathrm{Pr}\left(\left|\overline{{\mathbf{H}}(i,{\boldsymbol\eta})}.{\boldsymbol\nu}(i,\boldsymbol\alpha)\right|\ge\dfrac{2K}{2^{k/6}}\bigg|\left|B_1^{-1}(B_1(\boldsymbol\alpha))\right|\ge 2^{0.548kn}\right)
\\&\nonumber~~~~~~~~~~~~~~~~~~~~~~~~~~~~~~~~~~~\mathrm{Pr}\left(\left|B_1^{-1}(B_1(\boldsymbol\alpha))\right|\ge 2^{0.548kn}\right)\\&\nonumber+\mathrm{Pr}\left(\left|\overline{{\mathbf{H}}(i,{\boldsymbol\eta})}.{\boldsymbol\nu}(i,\boldsymbol\alpha)\right|\ge\dfrac{2K}{2^{k/6}}\bigg|\left|B_1^{-1}(B_1(\boldsymbol\alpha))\right|< 2^{0.548kn}\right)
\\&\nonumber
~~~~~~~~~~~~~~~~~~~~~~~~~~~~~~~~~~~\mathrm{Pr}\left(\left|B_1^{-1}(B_1(\boldsymbol\alpha))\right|< 2^{0.548kn}\right)\\&\nonumber\stackrel{(a)}\le  \mathrm{Pr}\left(\left|\overline{{\mathbf{H}}(i,{\boldsymbol\eta})}.{\boldsymbol\nu}(i,\boldsymbol\alpha)\right|\ge\dfrac{2K}{2^{k/6}}\bigg|\left|B_1^{-1}(B_1(\boldsymbol\alpha))\right|\ge 2^{0.548kn}\right)
\\&\nonumber~~~~~~~~~~~~~~~~~~~~~~~~~~~~~~~~
+\mathrm{Pr}\left(\left|B_1^{-1}(B_1(\boldsymbol\alpha))\right|< 2^{0.548kn}\right)\\&\nonumber\stackrel{(b)}\le
\dfrac{2^{0.544nk+k/3}}{2^{0.548kn}}+2^{-0.002kn}\\&\nonumber=2^{-0.002kn-0.002kn+k/3}+2^{-0.002kn} \\&\nonumber\stackrel{(c)}\le2^{-0.002kn}+2^{-0.002kn}\\&=2^{-0.002kn+1},
\end{align}
where $(a)$ follows from ignoring the probabilities $\mathrm{Pr}\left(\left|B_1^{-1}(B_1(\boldsymbol\alpha))\right|\ge 2^{0.548kn}\right)$ and
$$\mathrm{Pr}\left(\left|\overline{{\mathbf{H}}(i,{\boldsymbol\eta})}.{\boldsymbol\nu}(i,\boldsymbol\alpha)\right|\ge\dfrac{2K}{2^{k/6}}\bigg|\left|B_1^{-1}(B_1(\boldsymbol\alpha))\right|<2^{0.548kn}\right),$$
$(b)$ follows from i) using the upper-bound derived on $\mathrm{Pr}\left(\left|B_1^{-1}(B_1(\boldsymbol\alpha))\right|< 2^{0.548kn}\right)$ in Lemma~\ref{preimage} and ii) the fact that using Lemma~\ref{boundongoods}, we have 
\begin{align}\nonumber
&\mathrm{Pr}\left(\left|\overline{{\mathbf{H}}(i,{\boldsymbol\eta})}.{\boldsymbol\nu}(i,\boldsymbol\alpha)\right|\ge\dfrac{2K}{2^{k/6}}\bigg|\left|B_1^{-1}(B_1(\boldsymbol\alpha))\right|\ge 2^{0.548kn}\right)\le\dfrac{2^{0.544nk+k/3}}{2^{0.548kn}},
\end{align}
and $(c)$ follows because in practice, $n$ is sufficiently large such that $-0.002kn+k/3<0$. 

\bibliographystyle{IEEEtran}
\bibliography{Bibliography}

\end{document}